\documentclass[onecolumn]{IEEEtran}
\usepackage{algorithm}
\usepackage{algpseudocode}
\usepackage{times,amssymb,amsmath,amsfonts,nicefrac,float,accents,color,
graphicx,bbm,mathrsfs,amsthm,stmaryrd,soul}
\usepackage[noadjust]{cite}
\usepackage{float}
\usepackage{tikz}
\usetikzlibrary{calc}
\usepackage[mathscr]{eucal}
\usepackage{arydshln}
\usepackage{tikz}
\usepackage{array}
\usepackage{booktabs}

\usepackage{subfig}

\newtheorem{theorem}{Theorem}[section]
\newtheorem{lemma}[theorem]{Lemma}
\newtheorem{proposition}[theorem]{Proposition}
\newtheorem{corollary}[theorem]{Corollary}

\newtheorem{definition}[theorem]{Definition}

\newtheorem{remark}[theorem]{Remark}

\newcommand{\tb}{\textbf}

\begin{document}
\title{New bounds and constructions for constant weighted $X$-codes}
\author{Xiangliang Kong, Xin Wang and Gennian Ge

\thanks{X. Kong ({\tt 2160501011@cnu.edu.cn}) and  G. Ge ({\tt gnge@zju.edu.cn}) are with the School of Mathematical Sciences, Capital Normal University, Beijing 100048, China. The research of G. Ge was supported by the National Natural Science Foundation of China under Grant No. 11971325, National Key Research and Development Program of China under Grant Nos. 2020YFA0712100 and 2018YFA0704703, and Beijing Scholars Program.}
\thanks{X. Wang ({\tt xinw@suda.edu.cn}) is with the School of Mathematical Sciences, Soochow University, Suzhou 215006, China. The research of X. Wang was supported by the National Natural Science Foundation of China under Grant No. 11801392 and the Natural Science Foundation of Jiangsu Province under Grant No. BK20180833.}
}

\maketitle


\begin{abstract}
  As a crucial technique for integrated circuits (IC) test response compaction, $X$-compact employs a special kind of codes called $X$-codes for reliable compressions of the test response in the presence of unknown logic values ($X$s). From a combinatorial view point, Fujiwara and Colbourn \cite{FC2010} introduced an equivalent definition of $X$-codes and studied $X$-codes of small weights that have good detectability and $X$-tolerance.

  An $(m,n,d,x)$ $X$-code is an $m\times n$ binary matrix with column vectors as its codewords. The parameters $d,x$ correspond to the test quality of the code. In this paper, bounds and constructions for constant weighted $X$-codes are investigated. First, we obtain a general result on the maximum number of codewords $n$ for an $(m,n,d,x)$ $X$-code of weight $w$, and we further improve this lower bound for the case with $x=2$ and $w=3$ through the probabilistic method. Then, using tools from additive combinatorics and finite fields, we present some explicit constructions for constant weighted $X$-codes with $d=3,7$ and $x=2$, which are optimal for the case when $d=3, w=4$ and nearly optimal for the case when $d=3,w=3$. We also consider a special class of $X$-codes introduced in \cite{FC2010} and improve the best known lower bound on the maximum number of codewords for this kind of $X$-codes.
\end{abstract}


\IEEEpeerreviewmaketitle

\section{Introduction}\label{secintro}

Typical digital circuit testing applies test patterns to the circuit and observes the circuit's responses to the applied patterns. The observed response to a test pattern is compared with the expected response, and a chip in the circuit is determined to be defective if the comparison mismatches. With the development of the large scale integrated circuits (IC), although the comparison for each testing output is simple, the ever increasing amount of testing data costs much more time and space for processing. This leads to the requirement of more advanced test compression techniques \cite{MBKMPRW2003}. Since then, various related techniques have been studied such as automatic test pattern generation (ATPG) (see \cite{LC1985,SB1988,Hollmann1990,CZ1994} and the reference therein) and compression-based approaches (e.g., \cite{PLR2003,MLMP2005}).

Usually, voltages on signal lines in digital circuit system are interpreted as logic values $0$ or $1$. Based on this, by applying test patterns through fault-free simulations of the circuit, test engineers obtain the expected responses which are captured as $\{0,1\}$ vectors. Then, same test patterns are applied to the circuit and the circuit is declared to be defective if testing outputs are different from the expected responses. However, due to timing constraints, uninitialized memory elements, bus contention, inaccuracies of simulation models, etc (see Table 2 in \cite{MLMP2005}), for many digital systems, some simulated responses cannot be uniquely determined as $0$ or $1$ state. These unknown states are modeled as ``$X$'' states. In the presence of $X$s, the technique of $X$-compact was proposed in \cite{MK2004} as one of the compression-based approaches that have high reliability and error detection ability in actual digital systems.


$X$-compact uses $X$-codes as linear maps to compress test responses. An $(m,n,d,x)$ $X$-code is a set of $m$-dimensional $\{0,1\}$-vectors of size $n$ which can also be viewed as an $m\times n$ binary matrix with column vectors as codewords. The parameters $d,x$ correspond to the test quality of the code. The weight of a codeword $\mathbf{c}$ is the number of $1$s in $\mathbf{c}$. The value of $\frac{n}{m}$ is called the \emph{compaction ratio} and $X$-codes with large \emph{compaction ratios} are desirable for actual IC testing.

Let $M(m,d,x)$ be the maximum number $n$ of codewords for which there exists an $(m,n,d,x)$ $X$-code. To obtain $X$-codes with large compaction ratios, studies of the behavior of $M(m,d,x)$ are unavoidable. In \cite{FC2010}, based on a combinatorial approach, Fujiwara and Colbourn obtained a general lower bound $2^{{\frac{m}{2^{x+1}(d+x)}}}$ on $M(m,d,x)$ using probabilistic method (see Theorem 4.6, \cite{FC2010}). And this lower bound was further improved to $e^{{\frac{m-c_0}{e(x+1)(d+x-1)}}}$ by Tsuboda et al. in \cite{TFAV2018}.

In this paper, we focus on $X$-codes of constant weight. In coding theory, as an important class of codes, general constant weighted codes have been extensively studied for decades. They have played crucial roles in a number of engineering applications, including code-division multiple-access (CDMA) systems for optical fibers \cite{CSW1989}, protocol design for the collision channel without feedback \cite{GM1992}, automatic-repeat-request error-control systems \cite{WY1994}, and parallel asynchronous communication \cite{BB2000}. For the study of constant weighted codes, we recommend \cite{AVZ2000} and the reference therein.

As for $X$-codes, the weight of a codeword corresponds to the required fan-out of the $X$-compactor. Note that for an $X$-compactor, larger fan-in increases power requirements, area, and delay \cite{WH2003}. Due to the large amount of connections between $X$-compactors and inputs \cite{MK2004}, compactors with smaller fan-out inputs shall reduce fan-in values. From this point of view, codewords in $X$-codes are expected to have small weights. Therefore, $X$-codes of constant weight can be a good starting point for the study.

Firstly, in \cite{MLMP2005}, stochastic coding techniques are employed to design constant weighted $X$-compactors. For $x=1$, by viewing the matrix of an $(m,n,d,1)$ $X$-code as an incidence matrix of a graph, Wohl and Huisman \cite{WH2003} built a connection between this kind of $X$-codes with constant weight $2$ and graphs with girth at least $d+2$. For cases with multiple $X$s, given an $(m,n,d,x)$ $X$-code, Fujiwara and Colbourn \cite{FC2010} showed that a codeword of weight less than or equal to $x$ does not essentially contribute to the compaction ratio (see also \cite{LM2003}). Since then, aiming to achieve a large compaction ratio while minimizing the weight of each codeword, many works have been done about $(m,n,d,x)$ $X$-codes of constant weight $x+1$. Let $M_w(m,d,x)$ be the maximum number $n$ of codewords for which there exists an $(m,n,d,x)$ $X$-code of constant weight $w$. Using results from combinatorial design theory and superimposed codes, Fujiwara and Colbourn \cite{FC2010} proved that $M_3(m,d,2)=O(m^2)$ and $M_3(m,1,2)=\Theta(m^2)$. And they studied a special class of $(m,n,1,2)$ $X$-codes of constant weight $3$ with a property that boosts test quality when there are fewer unknowable bits than anticipated. In \cite{TF2018}, Tsunoda and Fujiwara proved that $M_3(m,d,2)=o(m^2)$ for $d\geq 4$ and they also improved the lower bound on the maximum number of codewords for the above special class of $(m,n,1,2)$ $X$-codes of constant weight $3$ introduced in \cite{FC2010}.


In this paper, based on the results from superimposed codes, additive combinatorics and extremal graph (hypergraph) theory, we obtain the following results:
\begin{itemize}
  \item General lower and upper bounds for $M_w(m,d,x)$.
  \item Explicit constructions for constant weighted $X$-codes with $d=3,7$ and $x=2$. These constructions further improve the general lower bound by providing a nearly optimal lower bound $m^{2-\varepsilon}$ for $M_{3}(m,3,2)$ and an optimal lower bound $c'm^{2}$ for $M_{4}(m,3,2)$, when $m$ is large enough.
  \item An improved lower bound for $X$-codes of constant weight $3$ with $x=2$ for any $d$:
        \begin{equation*}
         M_{3}(m,d,2)\geq c\cdot m^{\frac{9}{7}},
        \end{equation*}
        for some absolute constant $c>0$.
  \item An improvement of $({\log m})^{\frac{1}{5}}$ of the best known lower bound on the maximum number of codewords for the special class of $(m,n,1,2)$ $X$-codes of constant weight $3$ introduced in \cite{FC2010}. This improvement is also extended to the general case where higher error tolerances are required.
\end{itemize}

This paper is organised as follows: In Section II, we list some necessary notations and introduce the combinatorial requirements and the definitions of $X$-codes, we also include a lower bound for hypergraph independent sets preparing for proofs in Section IV. In Section III, we investigate the bounds and constructions for constant weighted $X$-codes. We prove a general result on $M_w(m,d,x)$ and a non-trivial lower bound on $M_3(m,d,2)$. We also present some explicit constructions for constant weighted $X$-codes with $d=3,7$ and $x=2$ based on the results from additive combinatorics and finite fields. In Section IV, we improve the lower bound on the maximum number of codewords for a special class of $(m,n,1,2)$ $X$-codes of constant weight $3$ and extend this result to a general case. In Section V, we conclude our work with some remarks.

\section{Preliminaries}\label{secpre}

\subsection{Notation}

We use the following notations throughout this paper.

\begin{itemize}
  \item Let $q$ be the power of a prime $p$, $\mathbb{F}_q$ be the finite field with $q$ elements, $\mathbb{F}_{q}^{n}$ be the vector space of dimension $n$ over $\mathbb{F}_q$.
  \item For any integer $n>0$, denote $[n]$ as the set of the first $n$ consecutive positive integers $\{1,2,\ldots,n\}$.
  \item For any vector $\tb{v}=(v_1,\cdots,v_n)\in \mathbb{F}_{q}^{n}$, let ${\rm supp}(\tb{v})=\{i\in [n] : v_i\neq 0\}$ and ${W}(\tb{v})=|{\rm supp}(v)|.$ For a set $S\subseteq[n],$ define $\tb{v}|_S=(v_{i_1},\ldots,v_{i_{|S|}})$, where $i_j\in S$ for $1\le j\le |S|$ and $1\le i_1<\cdots<i_{|S|}\le n$.
  \item For positive integer $k\geq 1$, a subset $P\subseteq\mathbb{F}_q$ of size $k$ is called an \emph{arithmetic progression of length $k$} if it has the form: $P=\{x+ia: x,a\in \mathbb{F}_q \text{ and } 0\leq i\leq k-1\}$. For simplicity, we denote \emph{$k$-AP} as the shortened form of \emph{arithmetic progression of length $k$}.
  \item  For functions $f=f(n)$ and $g=g(n)$, we use standard asymptotic notations $\Omega(\cdot)$, $\Theta(\cdot)$, $O(\cdot)$ and $o(\cdot)$ as $n\rightarrow \infty$:
    \begin{equation*}
      \begin{cases}
      f=O(g),~\text{if $\exists$ a constant}~c_1~\text{such that}~|f|\leq c_1|g|;\\
      f=\Omega(g),~\text{if $\exists$ a constant}~c_2~\text{such that}~|f|\geq c_2|g|;\\
      f=\Theta(g),~\text{if}~f=O(g)~\text{and}~f=\Omega(g);\\
      f=o(g),~\text{if}~\lim\limits_{n\rightarrow \infty}\frac{f}{g}=0.
    \end{cases}
    \end{equation*}
\end{itemize}

\subsection{$X$-Codes and Digital System Test Compaction}

To describe the behavior of unknown value $X$s, operations including addition (XOR) and multiplication (AND) for the $3$-valued logic system ($0$, $1$ and $X$) are formulated as $X$-$algebra$ by Fujiwara and Colbourn \cite{FC2010}: The $X$-algebra $\mathbb{X}_2=(\{0,1,X\},+,\cdot)$ over $\mathbb{F}_2$ is the set $\{0,1\}\subseteq \mathbb{F}_2$ and a third element $X$, equipped with two binary operations ``$+$'' (addition) and ``$\cdot$'' (multiplication) satisfying:
\begin{equation}\label{x_algebra}
\begin{cases}
1)\text{~For~}a,b\in \mathbb{F}_2, ~a+b \text{~and~} a\cdot b \text{~are performed in~} \mathbb{F}_2;\\
2)\text{~For~}a\in \mathbb{F}_2, ~a+X=X+a=X;\\
3)\text{~For the additive identity~} 0,~0\cdot X=X\cdot 0=0;\\
4)~1\cdot X=X\cdot 1=X.
\end{cases}
\end{equation}
Now, consider a circuit with response output $c=(c_1,\ldots,c_n)\in \{0,1,X\}^{n}$. Assume we have a test output $b=(b_1,\ldots,b_n)\in \{0,1\}^{n}$, based on the property of $X$-algebra, the $i_{th}$ bit is regarded as an error bit if and only if $b_i+c_i=1$.

For these testing and response outputs vectors, the $X$-compact technique is performed by right multiplying an $n\times m$ binary matrix $H$, where the arithmetics are carried out in $\mathbb{X}_2$. Denote $c'=(c'_1,\ldots,c'_m)=cH$ and $b'=(b'_1,\ldots,b'_m)=bH$ as the $X$-compacted outputs of the response vector $c$ and testing vector $b$ above. Similarly, the $i_{th}$ bit is regarded as an error bit if and only if $b'_i+c'_i=1$. Here, $H$ is called the $X$-compact matrix and the value of  $\frac{n}{m}$ is called the compaction ratio of $H$. To design $X$-compact matrices with large compaction ratio, $X$-codes were introduced in \cite{LM2003}. Roughly speaking, an $X$-code can be viewed as the set of row vectors of an $X$-compact matrix. To give the formal definition of $X$-codes, first, we introduce the following two operations on vectors.

Consider two $m$-dimensional vectors $\mathbf{s}_1=(s_1^{(1)},s_2^{(1)},\ldots,s_m^{(1)})$ and $\mathbf{s}_2=(s_1^{(2)},s_2^{(2)},\ldots,s_m^{(2)})$ where $s_{i}^{(j)}\in\mathbb{F}_2$. The \emph{addition} of $\mathbf{s}_1$ and $\mathbf{s}_2$ is bit-by-bit addition, denoted by $\mathbf{s}_1\oplus\mathbf{s}_2$; that is
\begin{equation*}
\mathbf{s}_1\oplus\mathbf{s}_2=(s_1^{(1)}+s_1^{(2)},s_2^{(1)}+s_2^{(2)},\ldots,s_m^{(1)}+s_m^{(2)}).
\end{equation*}
The \emph{superimposed sum} of $\mathbf{s}_1$ and $\mathbf{s}_2$, denoted by $\mathbf{s}_1 \vee\mathbf{s}_2$, is
\begin{equation*}
\mathbf{s}_1\vee\mathbf{s}_2=(s_1^{(1)}\vee s_1^{(2)},s_2^{(1)}\vee s_2^{(2)},\ldots,s_m^{(1)}\vee s_m^{(2)}),
\end{equation*}
where $s_i^{(j)}\vee s_k^{(l)}=0$ if $s_i^{(j)}=s_k^{(l)}=0$, otherwise $1$. And we say an $m$-dimensional vector $\mathbf{s}_1$ \emph{covers} an $m$-dimensional vector $\mathbf{s}_2$ if $\mathbf{s}_1\vee \mathbf{s}_2=\mathbf{s}_1$. For a finite set $S=\{\mathbf{s}_1,\ldots,\mathbf{s}_s\}$ of $m$-dimensional vectors, define
\begin{equation*}
\bigoplus S=\mathbf{s}_1\oplus\cdots\oplus\mathbf{s}_s,
\end{equation*}
and
\begin{equation*}
\bigvee S=\mathbf{s}_1\vee\cdots\vee\mathbf{s}_s.
\end{equation*}
When $s=1$, $\bigoplus S=\bigvee S=\{\mathbf{s}_1\}$, and when $S=\emptyset$, define $\bigoplus S=\bigvee S=\mathbf{0}$ (i.e. the zero vector).

\begin{definition}\cite{FC2010}\label{xcodedef}
Let $d$ be a positive integer and $x$ a nonnegative integer. An $(m,n,d,x)$ \emph{$X$-code} $\mathcal{X}=\{\mathbf{s}_1,\ldots,\mathbf{s}_n\}$ is a set of $m$-dimensional vectors over $\mathbb{F}_2$ such that $|\mathcal{X}|=n$ and
\begin{equation}\label{xcodedef1}
(\bigvee S_1)\vee(\bigoplus S_2)\neq \bigvee S_1
\end{equation}
for any pair of mutually disjoint subsets $S_1$ and $S_2$ of $\mathcal{X}$ with $|S_1|=x$ and $1\leq|S_2|\leq d$. A vector $\mathbf{s}_i\in \mathcal{X}$ is called a codeword. The weight of the codeword $\mathbf{s}_i$ is $|supp(\mathbf{s}_i)|$. The ratio $\frac{n}{m}$ is called the compaction ratio of $\mathcal{X}$.
\end{definition}

In view of $X$-compaction, the parameter $m$ of an $(m,n,d,x)$ $X$-code represents the size of the shrunk data, $n$ represents the number of bits in the raw response to be compressed at a time, $d$ corresponds to the discrepancy detecting ability and $x$ characterizes the unknowable bits tolerance. Generally speaking, as phrased in \cite{TFAV2018}, an $(m,n,d,x)$ $X$-code hashes the $n$-bit outputs from the circuit's test into $m$ bits while allowing for detecting the existence of up to $d$ bit-wise discrepancies between the actual outputs and correct responses even if up to $x$ bits of the correct behavior are unknowable to the tester.

From the definition above, when $x=0$, the codewords of an $(m,n,d,0)$ $X$-code actually form an $m\times n$ parity check matrix of a binary linear code of length $n$ with minimum distance $d+1$. Therefore, $(m,n,d,0)$ $X$-codes can be viewed as a special kind of traditional error-correcting codes.

When $x\geq1$ and $d\geq 2$, according to the definition, an $(m,n,d,x)$ $X$-code is also an $(m,n,d-1,x)$ $X$-code and an $(m,n,d,x-1)$ $X$-code. For the case when $x\geq 1$ and $d=1$, as pointed out in \cite{LM2003}, an $(m,n,1,x)$ $X$-code is equivalent to a \emph{$(1,x)$-superimposed code} of size $m\times n$.

\begin{definition}\cite{KS1964}\label{superimposedcodedef}
A \emph{$(1,x)$-superimposed code} of size $m\times n$ is an $m\times n$ matrix $S$ with entries in $\mathbb{F}_2$ such that no superimposed sum of any $x$ columns of $S$ covers any other column of $S$.
\end{definition}

Superimposed codes are also called \emph{cover-free families} and \emph{disjunct matrices}. These kinds of structures have been extensively studied in information theory, combinatorics and group testing. Thus, the bounds and constructions of $(1,x)$-superimposed codes can also be regarded as those for $(m,n,1,x)$ $X$-codes (see, for example, \cite{BFPR1984,DR1982,EFF1982,EFF1985,Furedi1996,HS1987,SW2000}).

\subsection{Independent sets in hypergraphs}

A hypergraph is a pair $(V,\mathcal{E})$, where $V$ is a finite set and $\mathcal{E}\subseteq 2^{V}$ is a family of subsets of $V$. The elements of $V$ are called vertices and the subsets in $\mathcal{E}$ are called hyperedges. We call $\mathcal{H}$ a $k$-uniform hypergraph, if all the hyperedges have the same size $k$, i.e., $\mathcal{E}\subseteq {V\choose k}$. For any vertex $v\in V$, we define the degree of $v$ to be the number of hyperedges containing $v$, denoted by $d(v)$. The maximum of the degrees of all the vertices is called the maximum degree of $\mathcal{H}$ and denoted by $\Delta(\mathcal{H})$.

An independent set of a hypergraph is a set of vertices containing no hyperedges and the independence number of a hypergraph is the size of its largest independent set. There are many results on the independence number of hypergraphs obtained through different methods (see \cite{AKPSS1982}, \cite{AKS1980}, \cite{DLR1995}, \cite{Kostochka2014}). Recall that a hypergraph $\mathcal{H}$ is linear if every pair of distinct hyperedges from $\mathcal{E}$ intersects in at most one
vertex. In this paper, we shall use the following version of the famous result of Ajtai et al. \cite{AKPSS1982} due to Duke et al. \cite{DLR1995} to derive some lower bounds on $M(m,d,x)$.

\begin{lemma}\cite{DLR1995}\label{DLR94}
Let $k\geq 3$ and let $\mathcal{H}$ be a $k$-uniform hypergraph with $\Delta(\mathcal{H})\leq D$. If $\mathcal{H}$ is linear, then
\begin{equation}\label{hypergraphind}
\alpha(\mathcal{H})\geq c\cdot |V| \cdot(\frac{\log D}{D})^{\frac{1}{k-1}},
\end{equation}
for some constant $c$ that depends only on $k$.
\end{lemma}

\section{Bounds and constructions of constant weighted $X$-codes}

In this section, we consider the bounds and constructions of constant weighted $X$-codes. This section is divided into three subsections. Section III-A includes a general result on the number of codewords of constant weighted $X$-codes from superimposed codes. Then in Section III-B, we give some explicit constructions for constant weighted $X$-codes with $d=3,7$ and $x=2$. And in Section III-C, we improve the general lower bound for $X$-codes of constant weight $3$ with $x=2$.

\subsection{General bounds from superimposed codes}

According to the definition, in \cite{MK2004}, the authors showed that an $(m,n,d,x)$ $X$-code is also an $(m,n,d+1,x-1)$ $X$-code. Note that for two binary vectors, their addition corresponds to the symmetric difference of their underlying sets and their superimposed sum corresponds to the union of their underlying sets. Therefore, by the equivalence between $X$-codes and superimposed codes, we have the following correspondence.
\begin{proposition}\label{prop0}
Let $d$ be a positive integer and $x$ be a nonnegative integer. A \emph{$(1,x+d-1)$-superimposed code} of size $m\times n$ is an $(m,n,d,x)$ $X$-code.
\end{proposition}
Denote $M_{w}(m,d,x)$ as the maximum number of codewords of an $(m,n,d,x)$ $X$-code of constant weight $w$. Since the restrictions for $X$-codes get more rigid with the growing of $d$, combined with the above proposition, we have
\begin{equation}\label{eq3}
M_{w}(m,1,x+d-1)\leq M_{w}(m,d,x)\leq M_{w}(m,1,x).
\end{equation}

In 1985, Erd\H{o}s et al. \cite{EFF1985} proved the following bounds on the maximum number of codewords of a $(1,x)$-superimposed code of constant weight $w$.
\begin{theorem}\cite{EFF1985}\label{upbforcwsc}
Denote $f_x(m,w)$ as the maximum number of columns of a $(1,x)$-superimposed code of constant weight $w$. Let $t=\lceil\frac{w}{x}\rceil$. Then, we have
\begin{equation*}\label{eq41}
\frac{{m\choose t}}{{w\choose t}^2}\leq f_x(m,w)\leq\frac{{m\choose t}}{{{w-1}\choose {t-1}}}.
\end{equation*}
Moreover, if we take $w=x(t-1)+1+\delta$ where $0\leq \delta<x$, then there exists a constant $m_0=m_0(w)$ such that for $m>m_0(w)$,
\begin{equation*}\label{eq2}
f_x(m,w)\geq(1-o(1))\frac{{{m-\delta}\choose t}}{{{w-\delta}\choose t}},
\end{equation*}
and $f_x(m,w)\leq\frac{{{m-\delta}\choose t}}{{{w-\delta}\choose t}}$ holds in the following cases:1) $\delta=0,1$; 2) $\delta<\frac{x}{2t^2}$; 3) $t=2$ and $\delta<\lceil\frac{2x}{3}\rceil$. Moreover, equality of the latter upper bound holds if and only if there exists a Steiner $t$-design $S(t,w-\delta,n-\delta)$.
\end{theorem}

According to this bound, by inequality (\ref{eq3}), we have the following immediate consequence:
\begin{theorem}\label{upbforcwxc}
Let $d,x$ be given positive integers, $w=x(t-1)+1+\delta$ where $0\leq \delta<x$, and $m_0=m_0(w)$ be the constant defined in Theorem \ref{upbforcwsc}. Then, for all $m\geq 1$,
\begin{equation}\label{eq51}
\frac{{m\choose {\lceil w/(x+d-1)\rceil}}}{{w\choose {\lceil w/(x+d-1)\rceil}}^2}\leq M_{w}(m,d,x)\leq\frac{{{m}\choose t}}{{{w-1}\choose {t-1}}}.
\end{equation}
And for $m>m_0(w)$,
\begin{equation}\label{eq52}
M_{w}(m,d,x)\geq (1-o(1))\frac{{{m}\choose {\lceil w/(x+d-1)\rceil}}}{{{w}\choose {\lceil w/(x+d-1)\rceil}}}
\end{equation}
and $M_{w}(m,d,x)\leq\frac{{{m-\delta}\choose t}}{{{w-\delta}\choose t}}$ holds in the following cases:1) $\delta=0,1$; 2) $\delta<\frac{x}{2t^2}$; 3) $t=2$ and $\delta<\lceil\frac{2x}{3}\rceil$.
\end{theorem}

In particular, for the case $x=2$, when $m>m_0(w)$, Theorem \ref{upbforcwxc} actually gives the following upper bound
\begin{equation}\label{eq6}
M_{w}(m,d,2)\leq
\begin{cases}
\frac{{{m-1}\choose w/2}}{{{w-1}\choose w/2}}, \text{when}~w~\text{is even};\\[3mm]
\frac{{{m}\choose (w+1)/2}}{{{w}\choose (w+1)/2}}, \text{when}~w~\text{is odd}.
\end{cases}
\end{equation}
According to the results from design theory, Fujiwara and Colbourn \cite{FC2010} proved the upper bound above is tight for the case $w=3$ and $d=1$, when there exists a corresponding Steiner triple system. Using the well-known \emph{graph removal lemma}, Tsunoda and Fujiwara \cite{TF2018} improved this upper bound on $M_3(m,d,2)$ to $o(m^2)$ for $d\geq 4$. So far as we know, for $d\geq 2$ and $x=2$, no upper or lower bounds better than these can be found in the literature.

\subsection{Explicit constructions of constant weighted $X$-codes}

\subsubsection{Constructions of constant weighted $X$-codes with $d=3$ and $x=2$}

In this part, we present two explicit constructions of constant weighted $X$-codes with $d=3$ and $x=2$, which provide asymptotically optimal lower bounds on $M_{4}(m,3,2)$ and nearly optimal lower bounds on $M_{3}(m,3,2)$, respectively.

\begin{itemize}
  \item Construction I : Let $p>w$ be a prime and $q$ be a power of $p$, take $w$ copies $X_1,X_2,\ldots,X_{w}$ of $\mathbb{F}_q$. Define
      \begin{align*}
      \mathcal{P}_1=\{&(x_1,x_2,\ldots,x_{w})\in \prod_{i=1}^{w}X_{i}:\\
      &x_1+(j-1)\cdot x_2+x_{j+1}=0~\text{for}~2\leq j\leq w-1\},
      \end{align*}
      as a family of $w$-tuples in $X_1\times \cdots\times X_{w}$. Clearly, $|\mathcal{P}_1|=q^2$. For $P_1\neq P_2\in \mathcal{P}_1$, we denote $P_1\cap P_2=\{i:P_1(i)=P_2(i),~1\leq i\leq w\}$ and define the indicator vector of ${P_{i}}$ as the concatenation of the $w$ indicator vectors of element $x_i$, i.e., $\mathbf{v}_{P_{i}}=(\mathbf{v}_{x_1},\mathbf{v}_{x_2},\ldots,\mathbf{v}_{x_{w}})$, where $\mathbf{v}_{x_i}$ is the indicator vector of element $x_i$ of length $q$. Let $\mathcal{C}_1$ be the set of all indicator vectors corresponding to $w$-tuples in $\mathcal{P}_1$.
\end{itemize}

\begin{theorem}\label{lwbforcwxc1}
For any $w\geq 4$ and prime $p>w$, let $q$ be a power of $p$, the code $\mathcal{C}_1$ from Construction I is a $(wq,q^{2},3,2)$ $X$-code of constant weight $w$.
\end{theorem}

\begin{proof}[Proof of Theorem \ref{lwbforcwxc1}]
From the definition, one can easily check that $|P_1\cap P_2|\leq 1$ for any two distinct $P_1,P_2\in\mathcal{P}_1$. Therefore, for integer $t\geq 1$, $\mathbf{v}_{P_1}\vee \mathbf{v}_{P_2}$ of any two distinct $P_1,P_2\in \mathcal{P}_1$ can cover at most $2t$ distinct ``1''s in $\mathbf{v}_{P_{3}}\oplus\cdots \oplus \mathbf{v}_{P_{t+2}}$ for other $t$ distinct $P_i$s in $\mathcal{P}_1$. Since $w\geq 4$ and $W(\bigoplus_{i=3}^{t+2}\mathbf{v}_{P_{i}})\geq t(w-t+1)$, this guarantees that the addition of any two or fewer vectors in $\mathcal{C}_1$ can not be covered by the superimposed sum of any other two vectors.

When $t=3$, assume there are $\{P_i\}_{i=1}^{5}$ such that $\mathbf{v}_{P_{3}}\oplus \mathbf{v}_{P_{4}} \oplus \mathbf{v}_{P_{5}}$ can be covered by $\mathbf{v}_{P_1}\vee \mathbf{v}_{P_2}$. Since $W(\mathbf{v}_{P_{3}}\oplus \mathbf{v}_{P_{4}} \oplus \mathbf{v}_{P_{5}})\geq 3(w-2)$, thus, we have $w=4$ and $W(\mathbf{v}_{P_3}\oplus \mathbf{v}_{P_4}\oplus \mathbf{v}_{P_5})=6$. Note that for $i\neq j$, $|P_i\cap P_j|\leq 1$. Thus, we have $|P_i\cap P_j|=1$ for $i\in \{1,2\}, j\in\{3,4,5\}$ and $|P_{j_1}\cap P_{j_2}|=1$ for $j_1,j_2\in \{3,4,5\}$.  Assume that $P_{j_1}\cap P_{j_2}=\theta_{j_1,j_2}$, $j_1,j_2\in \{3,4,5\}$. Since $w=4$, w.l.o.g., assume that $\theta_{3,4}=1$, $\theta_{3,5}=2$ and $\theta_{4,5}=3$. Therefore, we have $P_3(4),P_4(4),P_5(4)\in \{P_1(4)\}\cup\{P_2(4)\}$. By pigeonhole principle, w.l.o.g., we can assume that $P_3(4)=P_4(4)=P_1(4)$, this indicates that $|P_3\cap P_4|\geq 2$, a contradiction. Therefore, the addition of any three vectors in $\mathcal{C}_1$ can not be covered by the superimposed sum of any other two vectors. This indicates that $\mathcal{C}_1$ is a $(wq,|\mathcal{P}|,3,2)$ $X$-code of constant weight $w$.
\end{proof}

Actually, with the same spirit, there can be many other similar constructions providing the same bound. However, when $w=3$, this kind of constructions is no longer enough to guarantee the restrictions of being an $X$-code. For this case, we provide a new construction. First, we need the following lemma from \cite{EFR1986}.
\begin{lemma}\cite{EFR1986}\label{lem1}
For positive integers $w$ and $m$, there exists a set of positive integers $A\subseteq [m]$ of size
\begin{equation*}
|A|\geq \frac{m}{e^{c\log{w}\sqrt{\log{m}}}}
\end{equation*}
for some absolute constant $c$, such that $A$ contains no three terms of any arithmetic progressions of length $w$.
\end{lemma}

The specific construction of the set $A$ from Lemma \ref{lem1} can be regarded as an extension of the $3$-AP-free subset of $[m]$ given by Behrend \cite{B1946} and the detailed construction can be found in Section 5 of \cite{EFR1986}.

\begin{itemize}
  \item Construction II : Let $m_1=\lfloor\frac{m}{w}\rfloor$, $m_2=\lfloor\frac{m}{w^2}\rfloor$ and $A\subseteq [m_2]$ be the subset constructed from Lemma \ref{lem1} such that $A$ contains no three terms of any progressions of length $w$.
      Take $w$ copies $X_1,X_2,\ldots,X_{w}$ of $[m_1]$. Define
      \begin{align*}
      \mathcal{P}_2=\{&(x,x+a,\ldots,x+(w-1)a): \\
      &a\in A \text{~and~} x+(i-1)a\in X_i \text{~for~}1\leq i\leq w\},
      \end{align*}
      as a family of $w$-tuples in $X_1\times\cdots\times X_{w}$. Similarly, given $P_1\neq P_2\in \mathcal{P}_2$, denote $P_1\cap P_2=\{i:P_1(i)=P_2(i),~1\leq i\leq w\}$ and we define the indicator vector of ${P_{i}}$ as the concatenation of the $w$ indicator vectors of element $x_i$ together with an assistant zero vector, i.e., $\mathbf{v}_{P_{i}}=(\mathbf{v}_{x_1},\mathbf{v}_{x_2},\ldots,\mathbf{v}_{x_{w}},\mathbf{0})$, where $\mathbf{v}_{x_i}$ is the indicator vector of element $x_i$ of length $m_1$ and $\mathbf{0}$ is a zero vector of length $m-wm_1$. Let $\mathcal{C}_2$ be the set of all indicator vectors corresponding to $w$-tuples in $\mathcal{P}_2$.
\end{itemize}

\begin{theorem}\label{lwbforcwxc}
For any $\varepsilon>0$ and $w\geq 3$, there exists a constant $M=M(w,\varepsilon)>0$, such that for $m\geq M$, the code $\mathcal{C}_2$ from Construction II is an $(m,m^{2-\varepsilon},3,2)$ $X$-code of constant weight $w$.
\end{theorem}

\begin{proof}[Proof of Theorem \ref{lwbforcwxc}]
By the definition of $\mathcal{P}_2$, for $P_1\neq P_2\in \mathcal{P}_2$, we know that $|P_1\cap P_2|\leq 1$.
To proceed the proof, we need the following claim about the structure of $\mathcal{P}_2$.

\textbf{Claim.} $\mathcal{P}_2$ does not contain the following triple: $\{Q_1,Q_2,Q_3\}\subseteq \mathcal{P}_2$ satisfying that $Q_1\cap Q_2=\{\eta_1\}$, $Q_1\cap Q_3=\{\eta_2\}$, $Q_2\cap Q_3=\{\eta_3\}$, where $\eta_1,\eta_2,\eta_3\in\{1,2,\ldots,w\}$ are pairwise distinct.

\begin{proof}
Otherwise, assume that there are $\{Q_1,Q_2,Q_3\}\subseteq \mathcal{P}_2$ such that $Q_1\cap Q_2=\{\eta_1\}$, $Q_1\cap Q_3=\{\eta_2\}$, $Q_2\cap Q_3=\{\eta_3\}$ for three distinct ${\eta_1},{\eta_2},{\eta_3}$. By the definition of $\mathcal{P}_2$, for $1\leq i\leq 3$, we can assume that $Q_i=(x_i,x_i+a_i,\ldots,x_i+(w-1)a_i)$. Thus, we have
\begin{equation}\label{eq30}
\begin{cases}
Q_1(\eta_1)=x_1+(\eta_1-1)a_1=Q_2(\eta_1)=x_2+(\eta_1-1)a_2;\\
Q_1(\eta_2)=x_1+(\eta_2-1)a_1=Q_3(\eta_2)=x_3+(\eta_2-1)a_3;\\
Q_2(\eta_3)=x_2+(\eta_3-1)a_2=Q_3(\eta_3)=x_3+(\eta_3-1)a_3.
\end{cases}
\end{equation}
Combining these three equations in (\ref{eq30}) together, we have
\begin{equation*}
(\eta_2-\eta_1)a_1=(\eta_2-\eta_3)a_3+(\eta_3-\eta_1)a_2.
\end{equation*}
This means that $(\eta_2-\eta_3)(a_3-a_1)=(\eta_1-\eta_3)(a_2-a_1)$. Since $\eta_i$s are pairwise distinct, thus, both $\eta_2-\eta_3$ and $\eta_1-\eta_3$ are non-zero integers. Moreover, the distinctness of $Q_i$ also leads to $a_1,a_2,a_3$ being pairwise distinct. Thus, we have $a_3-a_1=\frac{\eta_1-\eta_3}{\eta_2-\eta_3}(a_2-a_1)$. W.l.o.g., assume that $\gcd(\eta_1-\eta_3,\eta_2-\eta_3)=1$. Then, take $D=\frac{a_2-a_1}{\eta_2-\eta_3}$, we have
\begin{equation*}
a_2=a_1+(\eta_2-\eta_3)D \text{ and } a_3=a_1+(\eta_1-\eta_3)D.
\end{equation*}
Since $\eta_1,\eta_2,\eta_3\in \{0,1,\ldots,w-1\}$, thus, $|\eta_i-\eta_j|<w$ for any $i\neq j\in [3]$. Therefore, $\{a_1,a_2,a_3\}\subseteq A$ are three pairwise distinct terms of a $w$-AP with common difference $D$. This contradicts the construction of $A$.
\end{proof}


With the help of this claim, next, for any two distinct $P_1,P_2\in \mathcal{P}_2$, we will verify that $\mathbf{v}_{P_1}\vee \mathbf{v}_{P_2}$ can not cover the addition of any at most three other vectors in $\mathcal{C}_2$.

First, since $|P_1\cap P_2|\leq 1$, $\mathbf{v}_{P_1}\vee \mathbf{v}_{P_2}$ can cover at most $2t$ distinct ``1''s in $\mathbf{v}_{P_{3}}\oplus\cdots \oplus \mathbf{v}_{P_{t+2}}$ for other $t$ distinct $P_is\in \mathcal{P}_1$. Note that $w\geq 3$, thus $\mathbf{v}_{P_1}\vee \mathbf{v}_{P_2}$ can not cover any other one vector in $\mathcal{C}_2$.

Second, assume that there exist other two distinct $P_3,P_4\in \mathcal{P}_2$ such that $\mathbf{v}_{P_3}\oplus \mathbf{v}_{P_4}$ is covered by $\mathbf{v}_{P_1}\vee \mathbf{v}_{P_2}$. When $w\geq 4$, we have $W(\mathbf{v}_{P_3}\oplus \mathbf{v}_{P_4})\geq 2(w-1)>4$. This indicates that one of the four intersections $|P_1\cap P_3|$, $|P_1\cap P_4|$, $|P_2\cap P_3|$, $|P_2\cap P_4|$ must be strictly larger than one, which is impossible. When $w=3$ and $W(\mathbf{v}_{P_3}\oplus \mathbf{v}_{P_4})=4$, this leads to $|P_1\cap P_3|=|P_1\cap P_4|=|P_2\cap P_3|=|P_2\cap P_4|=|P_3\cap P_4|=1$ and the intersection of any three of $P_1,P_2,P_3,P_4$ is an empty set. Thus, we can assume that $P_3\cap P_4=\{\theta_0\}$, $P_1\cap P_3=\{\theta_1\}$, $P_1\cap P_4=\{\theta_2\}$, where $\theta_0,\theta_1,\theta_2\in \{1,\ldots,w\}$ are pairwise distinct. This contradicts the claim above. Thus, $\mathbf{v}_{P_1}\vee \mathbf{v}_{P_2}$ can not cover the addition of any other two vectors in $\mathcal{C}_2$.


Now, assume that there exist other three distinct $\{P_3,P_4,P_5\}\subseteq \mathcal{P}_2$ such that $\mathbf{v}_{P_3}\oplus \mathbf{v}_{P_4}\oplus \mathbf{v}_{P_5}$ is covered by $\mathbf{v}_{P_1}\vee \mathbf{v}_{P_2}$. Since $\mathbf{v}_{P_1}\vee \mathbf{v}_{P_2}$ can cover at most $6$ distinct ``1''s in $\mathbf{v}_{P_3}\oplus \mathbf{v}_{P_4}\oplus \mathbf{v}_{P_5}$, thus, by $W(\mathbf{v}_{P_1}\oplus \mathbf{v}_{P_2}\oplus \mathbf{v}_{P_3})\geq 3(w-2)$, we can assume that $w\leq 4$.

When $w=3$, since $W(\mathbf{v}_{P_1}\vee \mathbf{v}_{P_2})\leq 6$ and $|P_i\cap P_j|\leq 1$ ($i\neq j\in \{3,4,5\}$), thus, either $W(\mathbf{v}_{P_3}\oplus \mathbf{v}_{P_4}\oplus \mathbf{v}_{P_5})=3$ or $W(\mathbf{v}_{P_3}\oplus \mathbf{v}_{P_4}\oplus \mathbf{v}_{P_5})=5$. For the case $W(\mathbf{v}_{P_3}\oplus \mathbf{v}_{P_4}\oplus \mathbf{v}_{P_5})=3$, we can assume that $P_3\cap P_4=\{\theta_0\}$, $P_3\cap P_5=\{\theta_1\}$, $P_4\cap P_5=\{\theta_2\}$, where $\theta_0,\theta_1,\theta_2\in \{1,\ldots,w\}$ are pairwise distinct. For the case $W(\mathbf{v}_{P_1}\oplus \mathbf{v}_{P_2}\oplus \mathbf{v}_{P_3})=5$, we can assume that $P_3\cap P_4=\{\theta_0\}$, $P_3\cap P_5=\{\theta_1\}$, $P_1\cap P_3=\{\theta_2\}$, $P_1\cap P_4=\{\theta_3\}$, where $\{\theta_{i}\}_{i=0}^{3}\subseteq \{1,\ldots,w\}$ are pairwise distinct. For both cases, we have three distinct $P_i$s pairwise intersecting at three distinct elements $\theta_j$s, which contradicts to the former claim.

When $w=4$, since $W(\mathbf{v}_{P_1}\vee \mathbf{v}_{P_2})\leq 8$ and $|P_i\cap P_j|\leq 1$ ($i\neq j\in \{3,4,5\}$), thus, either $W(\mathbf{v}_{P_3}\oplus \mathbf{v}_{P_4}\oplus \mathbf{v}_{P_5})=6$ or $W(\mathbf{v}_{P_3}\oplus \mathbf{v}_{P_4}\oplus \mathbf{v}_{P_5})=8$. For the case $W(\mathbf{v}_{P_3}\oplus \mathbf{v}_{P_4}\oplus \mathbf{v}_{P_5})=8$, we have $|P_i\cap P_j|>1$ for some $i\in\{1,2\}$ and $j\in\{3,4,5\}$, a contradiction. For the case $W(\mathbf{v}_{P_3}\oplus \mathbf{v}_{P_4}\oplus \mathbf{v}_{P_5})=6$, we can assume that $P_3\cap P_4=\{\theta_{34}\}$, $P_3\cap P_5=\{\theta_{35}\}$, $P_4\cap P_5=\{\theta_{45}\}$ and $P_i\cap P_j=\{\theta_{ij}\}$ for each $i\in\{1,2\},j\in\{3,4,5\}$, where $\theta_{ij}\in \{1,2,\ldots,w\}$ are pairwise distinct. This also leads to three distinct $P_i$s pairwise intersecting at three distinct elements $\theta_{ij}$s, which contradicts the construction of $\mathcal{P}_2$.

In conclusion, the addition of any three or fewer vectors in $\mathcal{C}_2$ can not be covered by the superimposed sum of any other two vectors. Since $|A|\geq \frac{m_2}{e^{c\log{w}\sqrt{\log{m_2}}}}$ for some $c>0$, we have $|\mathcal{P}_2|\geq m_2|A|\geq m^{2-\varepsilon}$ for every $\varepsilon>0$ and $m\geq M$, therefore, $\mathcal{C}_2$ is the desired $(m,m^{2-\varepsilon},3,2)$ $X$-code of constant weight $w$.
\end{proof}


\begin{remark}
According to the upper bound given by (\ref{eq6}), we have
\begin{equation*}
\begin{cases}
M_{3}(m,3,2)\leq\frac{m(m-1)}{6}, \\
M_{4}(m,3,2)\leq\frac{(m-1)(m-2)}{6}.
\end{cases}
\end{equation*}
Therefore, for the case $w=3$, the lower bound $m^{2-\varepsilon}$ from Theorem \ref{lwbforcwxc} is nearly optimal; and for the case $w=4$, the lower bound $c'm^{2}$ from Theorem \ref{lwbforcwxc1} is optimal, regardless of a constant factor. For cases when $w\geq 9$, (\ref{eq52}) in Theorem \ref{upbforcwxc} provides better lower bounds $(1-o(1))\frac{{m\choose {\lceil w/4\rceil}}}{{w\choose {\lceil w/4\rceil}}}$, but the gaps between the upper bounds and the lower bounds are still quite large.

It is also worth noting that, the construction from Theorem \ref{lwbforcwxc} was originally proposed by Erd\H{o}s et al. \cite{EFR1986} to construct $w$-uniform hypergraphs
on $m$ vertices such that no $3w-3$ vertices span $3$ or more hyperedges. This kind of hypergraphs is a special kind of sparse hypergraphs which will be discussed later in Section III.D.
\end{remark}

\subsubsection{Construction of constant weighted $X$-codes with $d=7$ and $x=2$}

Before we present the construction, we shall prove a proposition which establishes a connection between constant weighted $X$-codes with $d=7, x=2$ and uniform hypergraphs of girth five.

Given a $k$-uniform hypergraph $\mathcal{H}=(V,\mathcal{E})$ and a positive integer $l\geq 2$, a cycle of length $l$ in $\mathcal{H}$ ($l$-cycle in short), denoted by $\mathbb{C}_l$, is an alternating sequence of distinct vertices and hyperedges of the form: $v_1,E_1,v_2,E_2,\ldots,v_l,E_l,v_1$, such that $\{v_i,v_{i+1}\}\subseteq E_{i}$ for each $i\in\{1,2,\ldots,l\}$ and $\{v_l,v_1\}\subseteq E_l$. A linear path of length $l$ ($l$-path in short), denoted by $\mathbb{P}_l$, is an alternating sequence of distinct vertices and hyperedges of the form: $E_1,v_2,E_2,v_3,\ldots,v_l,E_l$, such that $E_i\cap E_{i+1}=\{v_{i+1}\}$ for each $i$ and $E_i\cap E_j=\emptyset$ whenever $|j-i|>1$. And the \emph{girth} of hypergraph $\mathcal{H}$ is the minimum length of a cycle in $\mathcal{H}$.

\begin{proposition}\label{xcandhgofg}
Let $w\geq3$ be a positive integer. For any $w$-uniform hypergraph $\mathcal{H}=(V,\mathcal{E})$ of girth at least $5$, the set of all the indicator vectors of hyperedges in $\mathcal{E}$ forms a $(|V|,|\mathcal{E}|,7,2)$ $X$-code of constant weight $w$.
\end{proposition}

\begin{proof}[Proof of Proposition \ref{xcandhgofg}]
First, note that the girth of $\mathcal{H}$ is at least $5$, we know that $\mathcal{H}$ is a linear hypergraph, i.e., $|E_1\cap E_2|\leq 1$ for any $E_1,E_2\in \mathcal{E}$. Hence, if we denote $\mathbf{v}_{E_i}$ as the indicator vector of hyperedge $E_i$, then for any $\{E_{1},\ldots,E_{7}\}\subseteq{\mathcal{E}}$ and any $s$-subset $I_s\subseteq [7]$ with $1\leq s\leq 7$, we have
\begin{equation*}
W(\bigoplus_{i\in I_s}\mathbf{v}_{E_{i}})\geq s\cdot(w-s+1).
\end{equation*}
Moreover, for every $E\in \mathcal{E}$, $\mathbf{v}_{E}$ can't be covered by the superimposed sum of the indicator vectors of any other two edges in $\mathcal{E}$. For each $2\leq s\leq 7$ and an $s$-subset $I_s\subseteq [7]$, consider the subhypergraph spanned by $\{E_i\}_{i\in{I_s}}$, we denote $V_0(I_s)$ as the set of vertices with even degree in this subhypergraph and $V_1(I_s)$ as the set of vertices with odd degree in this subhypergraph.

Let $\mathcal{C}$ be the set of indicator vectors of all edges in $\mathcal{E}$, according to the restrictions of the $(|V|,|\mathcal{E}|,7,2)$ $X$-code, our proof is divided into the following three parts.

\textbf{Case 1.} Assume that there exist $\{E_i\}_{i=1}^{9}\subseteq \mathcal{E}$ such that $\bigoplus_{i\in [7]}\mathbf{v}_{E_i}$ is covered by $\mathbf{v}_{E_8}\vee \mathbf{v}_{E_9}$. 

When the length of the longest linear path in the subhypergraph formed by $\{E_i\}_{i=1}^{7}$ is at most $3$, consider a longest linear path $\mathbb{P}^{(7)}$ formed by edges $\{E_i\}_{i\in S}$ for some subset $S\subseteq [7]$ of size at most $3$. Since $\mathcal{H}$ has girth at least $5$, therefore, by the maximality of $\mathbb{P}^{(7)}$, the starting edge $E_{i_s}$ and the ending edge $E_{i_e}$ of $\mathbb{P}^{(7)}$ are disjoint with all edges in $\{E_i\}_{i\in [7]\setminus S}$. Therefore, by $w\geq 3$, we have $|V_1(S)\cap E_{i_s}|,|V_1(S)\cap E_{i_e}|\geq 2$. Note that $V_{1}(S)$ is covered by $\mathbf{v}_{E_8}\vee \mathbf{v}_{E_9}$. This forces $E_8$ (or $E_9$) together with $\mathbb{P}^{(7)}$ to form a cycle of length at most $4$, which contradicts the requirement of $\mathcal{H}$ having girth at least $5$.

When the length of the longest linear path in the subhypergraph formed by $\{E_i\}_{i=1}^{7}$ is at least $4$, consider a linear $3$-path $\mathbb{P}_{1}^{(7)}$ formed by edges $\{E_i\}_{i\in S_1}$ for some $3$-set $S_1\subseteq [7]$. The vector $\bigoplus_{i\in S_1}\mathbf{v}_{E_i}$ has weight
\begin{equation*}
W(\bigoplus_{i\in S_1}\mathbf{v}_{E_i})= 3(w-2)+2.
\end{equation*}
As $\mathcal{H}$ has girth at least $5$, for each $i\in [9]\setminus S_1$, $\mathbf{v}_{E_i}$ has at most one coordinate with value ``$1$'' agreeing with $\bigoplus_{i\in S_1}\mathbf{v}_{E_i}$. Then, the assumption that $\bigoplus_{i\in [7]}\mathbf{v}_{E_i}$ being covered by $\mathbf{v}_{E_8}\vee\mathbf{v}_{E_9}$ leads to $3(w-2)+2\leq 6$. Therefore, we have $w\leq 3$.

Take $I_s$ as $[7]$, then the assumption indicates that $V_1([7])\subseteq E_8\cup E_9$. Since $w=3$, we have $|V_1([7])|\leq 6$. One can easily check this only holds when the configuration formed by $\{E_1,\ldots,E_{7}\}$ is isomorphic to the subhypergraph shown in Fig. 1. Since there are $5$ distinct vertices with odd degree in this configuration, the assumption that $\bigoplus_{i\in [7]}\mathbf{v}_{E_i}$ being covered by $\mathbf{v}_{E_8}\vee\mathbf{v}_{E_9}$ forces that $E_8$ forms a linear cycle of length at most $4$ with $2$ or $3$ distinct hyperedges in $\{E_1,\ldots,E_{7}\}$, a contradiction.

\begin{figure*}
  \centering
  \includegraphics[width=5cm]{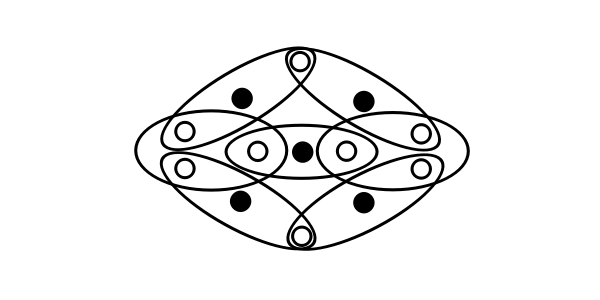}
  \caption{Subhypergraph formed by $\{E_1,\ldots,E_{7}\}$ with $5$ odd vertices, where vertices with odd degree are denoted as ``$\bullet$'' and vertices with even degree are denoted as ``$\circ$''. }\label{pic04}
\end{figure*}



\textbf{Case 2.} Assume that there exist $\{E_i\}_{i=1}^{8}\subseteq \mathcal{E}$ such that $\bigoplus_{i\in [6]}\mathbf{v}_{E_i}$ is covered by $\mathbf{v}_{E_7}\vee \mathbf{v}_{E_8}$. 

Similar to the analysis in Case 1, when the length of the longest linear path in the subhypergraph formed by $\{E_i\}_{i=1}^{6}$ is at most $3$, the assumption that $\bigoplus_{i\in [6]}\mathbf{v}_{E_i}$ being covered by $\mathbf{v}_{E_7}\vee \mathbf{v}_{E_8}$ forces $E_7$ (or $E_8$) together with one of the longest linear path to form a cycle of length at most $4$, a contradiction.

When the length of the longest linear path in the subhypergraph formed by $\{E_i\}_{i=1}^{6}$ is at least $4$, consider a $3$-path $\mathbb{P}^{(6)}$ formed by $\{E_i\}_{i\in S_2}$ for some $3$-subset $S_2\subseteq [6]$, we have
\begin{equation*}
W(\bigoplus_{i\in S_2}\mathbf{v}_{E_i})= 3(w-2)+2.
\end{equation*}
As $\mathcal{H}$ has girth at least $5$, for each $i\in [8]\setminus S_2$, $\mathbf{v}_{E_i}$ has at most one coordinate with value ``$1$'' agreeing with $\bigoplus_{i\in S_2}\mathbf{v}_{E_i}$. Therefore, the assumption that $\bigoplus_{i\in [6]}\mathbf{v}_{E_i}$ being covered by $\mathbf{v}_{E_7}\vee \mathbf{v}_{E_8}$ implies that $3(w-2)+2\leq 5$. Thus, we have $w\leq 3$.

Take $I_s$ as $[6]$, then we have $|V_1([6])|\leq 6$. One can easily check this only holds when the configuration formed by $\{E_1,\ldots,E_{6}\}$ is isomorphic to one of the subhypergraphs shown in Fig. 2. Since there are $6$ distinct vertices with odd degree in this configuration, the assumption that $\bigoplus_{i\in [6]}\mathbf{v}_{E_i}$ being covered by $\mathbf{v}_{E_7}\vee\mathbf{v}_{E_8}$ forces that $E_8$ forms a linear cycle of length at most $4$ with $2$ or $3$ distinct hyperedges in $\{E_1,\ldots,E_{6}\}$, a contradiction.

\begin{figure*}
  \centering
    \scalebox{1}[1]{\includegraphics[width=3.75cm]{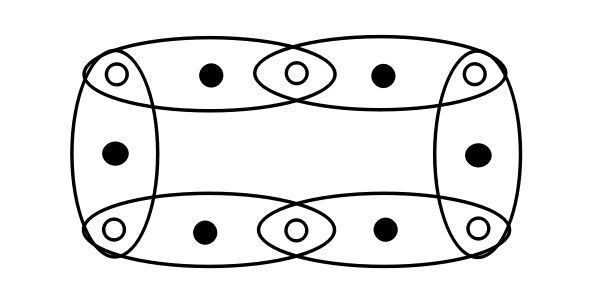}}\quad
    \scalebox{1}[1.05]{\includegraphics[width=3.75cm]{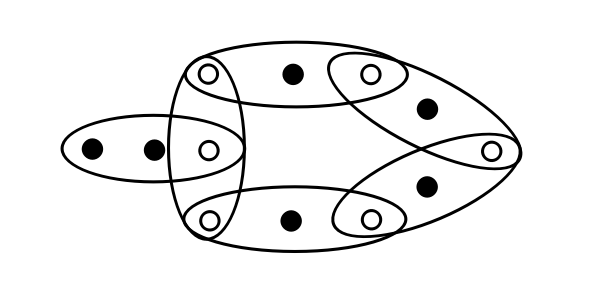}}\quad
    \caption{Subhypergraphs formed by $\{E_1,\ldots,E_{6}\} $ with $6$ odd vertices, where vertices with odd degree are denoted as ``$\bullet$'' and vertices with even degree are denoted as ``$\circ$''.}\label{pic03}
\end{figure*}

\textbf{Case 3.} For each $4\leq l\leq 7$, assume that there exist $\{E_i\}_{i=1}^{l}\subseteq \mathcal{E}$ such that $\bigoplus_{i\in [l-2]}\mathbf{v}_{E_i}$ is covered by $\mathbf{v}_{E_{l-1}}\vee \mathbf{v}_{E_{l}}$. Similar to the analysis in Case 1 and Case 2, we only have to consider the case when the length of the longest linear path in the configuration formed by $\{E_i\}_{i=1}^{l-2}$ is at least $4$.

Let $\mathbb{P}^{(l-2)}$ be a $3$-path in this subgraph formed by $\{E_i\}_{i\subseteq S_3}$ for some $3$-subset $S_3\subseteq [l-2]$, we have
\begin{equation*}
W(\bigoplus_{i\in S_3}\mathbf{v}_{E_i})= 3(w-2)+2.
\end{equation*}
Again, by the girth restriction of $\mathcal{H}$, for each $i\in [l]\setminus S_3$, $\mathbf{v}_{E_i}$ has at most one coordinate with value ``$1$'' agreeing with $\bigoplus_{i\in S_3}\mathbf{v}_{E_i}$. Therefore, the assumption above indicates that $3(w-2)+2\leq (l-3)$, which leads to $w\leq 2$. This contradicts the fact that $w\geq 3$.

In conclusion, the addition of any seven or fewer distinct vectors in $\mathcal{C}$ can not be covered by the superimposed sum of any other two vectors in $\mathcal{C}$. Therefore, $\mathcal{C}$ is a $(|V|,|\mathcal{E}|,7,2)$ $X$-code of constant weight $w$.
\end{proof}

Based on a construction of $3$-uniform hypergraphs of girth at least five in \cite{LV2003}, by Proposition \ref{xcandhgofg}, we have the following result.

\begin{theorem}\label{xcofcw3d4x2}
For any odd prime power $q$, there exists a $(q(q-1),{q\choose 3},7,2)$ $X$-code of constant weight $3$.
\end{theorem}

\begin{proof}[Proof of Theorem \ref{xcofcw3d4x2}]
For any odd prime power $q$, consider the finite field $\mathbb{F}_q$, let $C_q$ denote the set of points on the curve $2x_2={x_1}^2$, where $(x_1,x_2)\in \mathbb{F}_q^{2}$.

Define a hypergraph $\mathcal{G}_q$ with vertex set $V(\mathcal{G}_q)=\mathbb{F}_q^2\setminus C_q$. Three distinct vertices $\mathbf{a}=(a_1,a_2)$, $\mathbf{b}=(b_1,b_2)$ and $\mathbf{c}=(c_1,c_2)$ form a hyperedge $\{\mathbf{a},\mathbf{b},\mathbf{c}\}$ in $\mathcal{G}_q$ if and only if the following three equations hold:
\begin{equation*}
\begin{cases}
a_2+b_2=a_1b_1;\\
b_2+c_2=b_1c_1;\\
c_2+a_2=c_1a_1.
\end{cases}
\end{equation*}

As claimed in \cite{LV2003} (see the Remark on page 9 in \cite{LV2003}), $\mathcal{G}_q$ has girth at least five. Clearly, there are ${q\choose 3}$ choices for distinct numbers $a_1$, $b_1$ and $c_1$, and each choice uniquely specifies $a_2$, $b_2$ and $c_2$ satisfying the above three equations. This indicates that any two choices of the triple $\{a_1,b_1,c_1\}$ being the same will lead to identical corresponding hyperedges. Therefore, the number of hyperedges in $\mathcal{G}_q$ is precisely ${q\choose 3}$. By Proposition \ref{xcandhgofg}, we obtain a $(q(q-1),{q\choose 3},7,2)$ $X$-code of constant weight $3$.
\end{proof}

\begin{remark}\label{remark3}
The construction from Theorem \ref{xcofcw3d4x2} actually gives a lower bound on $M_3(m,7,2)$ of the form
\begin{equation*}
M_3(m,7,2)=\Omega (m^{\frac{3}{2}}),
\end{equation*}
for sufficiently large $m$. This is better than the lower bound given by (\ref{eq52}) in Theorem \ref{upbforcwxc} in this case, however, compared to the upper bound $o(m^2)$ given by Tsunoda and Fujiwara \cite{TF2018}, there is still a gap.

Unfortunately, this construction can not be extended to obtain general constant weighted $X$-codes. But at least, together with Proposition \ref{xcandhgofg}, it provides a way for constructing large constant weighted $X$-codes with $d=7$ and $x=2$.
\end{remark}

\subsection{An improved lower bound for $X$-codes of constant weight $3$ with $x=2$}

Notice that when taking $w=x+1$ in Theorem \ref{upbforcwxc}, the general lower bound given by (\ref{eq52}) is only a linear function of $m$ for $d\geq 2$. Through an elaborate analysis of the connection between a special kind of $3$-uniform hypergraphs and $X$-codes of constant weight $3$, we prove the following theorem, which improves this lower bound to $\Omega(m^{\frac{9}{7}})$.

\begin{theorem}\label{lbxcwithw3x2}
For any positive integer $d\geq 8$ and sufficiently large $m$, there exists an $(m,c\cdot m^{\frac{9}{7}},d,2)$ $X$-code of constant weight $3$, where $c>0$ is an absolute constant.
\end{theorem}

In graph theory, a $k$-uniform hypergraph $\mathcal{H}$ is called $\mathcal{G}_k(v,e)$-free if the union of any $e$ distinct hyperedges contains at least $v+1$ vertices. These kinds of hypergraphs are called \emph{sparse hypergraphs}. They are important structures in extremal graph theory and have been well-studied since 1970s (see \cite{AS06,Sudakov2010,Keevash2011,GS2017} and the reference therein). Before we present the proof of Theorem \ref{lbxcwithw3x2}, we need the following lemma.

\begin{lemma}\label{xcandshgofg}
For any $3$-uniform hypergraph $\mathcal{H}=(V,\mathcal{E})$ that is simultaneously $\mathcal{G}_3(2s,s)$-free for each $2\leq s\leq 4$ and $\mathcal{G}_3(\lceil\frac{3s-1}{2}\rceil+3,s)$-free for each $8\leq s\leq d$, the set of all the indicator vectors of hyperedges in $\mathcal{E}$ forms a $(|V|,|\mathcal{E}|,d,2)$ $X$-code of constant weight $3$.
\end{lemma}

\begin{proof}[Proof of Lemma \ref{xcandshgofg}]
Consider a $3$-uniform hypergraph $\mathcal{H}_0=(V_0,\mathcal{E}_0)$ that is simultaneously $\mathcal{G}_3(2s,s)$-free for each $2\leq s\leq 4$ and $\mathcal{G}_3(\lceil\frac{3s-1}{2}\rceil+3,s)$-free for each $8\leq s\leq d$. Since $\mathcal{H}_0$ is $\mathcal{G}_3(2s,s)$-free for each $2\leq s\leq 4$, we know that the girth of $\mathcal{H}_0$ is at least $5$. From the result of Proposition \ref{xcandhgofg}, the set of all the indicator vectors $\mathcal{C}(\mathcal{H}_0)$ corresponding to $\mathcal{E}_0$ already forms a $(|V|,|\mathcal{E}|,7,2)$ $X$-code of constant weight $3$. Therefore, we only have to show that the addition of any $s~(8\leq s\leq d)$ distinct indicator vectors in $\mathcal{C}(\mathcal{H}_0)$ can not be covered by the superimposed sum of any other two indicator vectors in $\mathcal{C}(\mathcal{H}_0)$.

For each $e\in \mathcal{E}_0$, denote $\mathbf{v}_{e}$ as the indicator vector of $e$. For each integer $8\leq s\leq d$, consider $s$ distinct hyperedges $\{e_1,\ldots,e_{s}\}$ in $\mathcal{E}$. Assume that there exist two other hyperedges $f_1$ and $f_2$, such that $\mathbf{v}_{e_1}\oplus\cdots\oplus\mathbf{v}_{e_{s}}$ can be covered by $\mathbf{v}_{f_1}\vee\mathbf{v}_{f_2}$. Denote $V_0$ as the set of vertices in $\bigcup_{i=1}^{s}e_{i}$ that are contained in even number of hyperedges in $\{e_1,\ldots,e_{s}\}$ and $V_1$ as the set of vertices in $\bigcup_{i=1}^{s}e_{i}$ that are contained in odd number of hyperedges in $\{e_1,\ldots,e_{s}\}$. Then the assumption indicates that $V_1\subseteq f_1\cup f_2$. Since $\mathcal{H}$ is a $3$-uniform hypergraph, we have
\begin{equation}\label{8}
|V_1|\leq 6 \text{~and~} 2|V_0|+|V_1|\leq 3s.
\end{equation}

Now, for a fixed integer $8\leq s_0\leq d$, according to inequality (\ref{8}) and the parity of $s_0$, we have
\begin{equation*}
|\bigcup_{i=1}^{s_0}e_{i}|=|V_0|+|V_1|\leq \lceil\frac{3s_0-1}{2}\rceil+3.
\end{equation*}
This implies that these $s_0$ distinct hyperedges $\{e_1,\ldots,e_{s_0}\}$ are spanned by at most $\lceil\frac{3s_0-1}{2}\rceil+3$ distinct vertices in $\mathcal{H}_0$, which contradicts the condition that $\mathcal{H}_0$ is $\mathcal{G}_3(\lceil\frac{3s-1}{2}\rceil+3,s)$-free for each $8\leq s\leq d$. Thus, for each $8\leq s\leq d$, the addition of any distinct $s$ indicator vectors in $\mathcal{C}(\mathcal{H}_0)$ can not be covered by the superimposed sum of other $2$ indicator vectors. Therefore, combined with former analysis, the set of all the indicator vectors in $\mathcal{C}(\mathcal{H}_0)$ forms a $(|V|,|\mathcal{E}|,d,2)$ $X$-code of constant weight $3$.
\end{proof}

%

Now, we present the proof of Theorem \ref{lbxcwithw3x2}.

\begin{proof}[Proof of Theorem \ref{lbxcwithw3x2}]
By Lemma \ref{xcandshgofg}, we only need to construct a $3$-uniform hypergraph $\mathcal{H}_0$ that is simultaneously $\mathcal{G}_3(2s,s)$-free for each $2\leq s\leq 4$ and $\mathcal{G}_3(\lceil\frac{3s-1}{2}\rceil+3,s)$-free for each $8\leq s\leq d$ with $\Omega(m^{\frac{9}{7}})$ hyperedges.

%

Let $V$ be a finite set of points and $|V|=m$, take a subset $\mathcal{B}$ of triples by picking elements of ${V\choose 3}$ uniformly and independently at random with probability $p$. Then we have
\begin{equation*}
\mathbb{E}[|\mathcal{B}|]=p\cdot{|V|\choose 3}.
\end{equation*}

For each $2\leq s\leq 4$, denote $D_s$ as the set of $s$-subsets in $\mathcal{B}$ that are spanned by at most $2s$ points in $V$, i.e., for each $\{B_1,\ldots,B_s\}\in D_s\subseteq {\mathcal{B}\choose s}$, $|\bigcup_{i=1}^s{B_i}|\leq 2s$. Then we have
\begin{equation*}
p^s\cdot{|V|\choose 2s}\leq\mathbb{E}[|D_s|]\leq {2s\choose 3}^{s}\cdot p^s\cdot{|V|\choose 2s},
\end{equation*}
for each $2\leq s\leq 4$.

For each $8\leq s\leq d$, denote $D_s$ as the set of $s$-subsets in $\mathcal{B}$ that are spanned by at most $\lceil\frac{3s-1}{2}\rceil+3$ points in $V$, i.e., for each $\{B_1,\ldots,B_s\}\in D_s\subseteq {\mathcal{B}\choose s}$, $|\bigcup_{i=1}^s{B_i}|\leq \lceil\frac{3s-1}{2}\rceil+3$. Then we have
\begin{align*}
p^s\cdot{|V|\choose \lceil\frac{3s-1}{2}\rceil+3} &\leq\mathbb{E}[|D_s|]\\
&\leq {{\lceil\frac{3s-1}{2}\rceil+3}\choose 3}^{s}\cdot p^s\cdot{|V|\choose \lceil\frac{3s-1}{2}\rceil+3},
\end{align*}
for each $8\leq s\leq d$.

By deleting at most one triple from each $s$-subset in $D_s$, for $2\leq s\leq 4$ and $8\leq s\leq d$, the remaining triples form a $3$-uniform hypergraph that is simultaneously $\mathcal{G}_3(2s,s)$-free for each $2\leq s\leq 4$ and $\mathcal{G}_3(\lceil\frac{3s-1}{2}\rceil+3,s)$-free for each $8\leq s\leq d$. Now, take $p=\frac{1}{30}\cdot m^{-\frac{12}{7}}$ and $0\leq c\leq \frac{1}{30}$. For $m$ sufficiently large, we have $\mathbb{E}[|D_s|]=o(\mathbb{E}[|\mathcal{B}|])$ for $2\leq s\leq d, s\neq 8$. Therefore, by the linearity of expectation, we have
\begin{align*}
\mathbb{E}[|\mathcal{B}|-\sum_{s=2}^{d}|D_s|]&\geq p\cdot{|V|\choose 3}\cdot(1-o(1))
-\frac{(455)^{8}}{14!}\cdot p^8\cdot|V|^{15}\\
&\geq c\cdot m^{\frac{9}{7}}.
\end{align*}
Therefore, with positive probability, there exists a $3$-uniform hypergraph $\mathcal{H}$ that is simultaneously $\mathcal{G}_3(2s,s)$-free for each $2\leq s\leq 4$ and $\mathcal{G}_3(\lceil\frac{3s}{2}\rceil+3,s)$-free for each $8\leq s\leq d$ with vertex set $V$ and $c\cdot m^{\frac{9}{7}}$ hyperedges. This completes the proof.
\end{proof}

\section{$r$-even-free triple packings and $X$-codes with higher error tolerance}

To construct $X$-codes with $x=2$ and weight $3$, Fujiwara and Colbourn \cite{FC2010} introduced the notion of $r$-even-free triple packing, which was further studied in \cite{TF2018}. In this section, by obtaining an existence result of the corresponding $6$-even-free triple packing, we prove a lower bound on the maximum number of codewords of an $(m,n,1,2)$ $X$-code of constant weight $3$ which can detect up to three erroneous bits if there is only one $X$ in the raw response data and up to six erroneous bits if there is no $X$, this improves the lower bound given in \cite{TF2018}. And we also extend this lower bound to a general case.


A \emph{triple packing} of order $v$ is a set system $(V,\mathcal{B})$ such that $\mathcal{B}$ is a family of triples of a finite set $V$ and any pair of elements of $V$ appears in $\mathcal{B}$ at most once. Given a triple packing $(V,\mathcal{B})$, we call subset $\mathcal{C}$ in $\mathcal{B}$ an $i$-configuration if $|\mathcal{C}|=i$. A configuration $\mathcal{C}$ is \emph{even} if for every vertex $v\in V$ appearing in $\mathcal{C}$, the number $|\{B:v\in{B}\in{\mathcal{C}}\}|$ of triples containing $v$ is even. And a triple packing $(V,\mathcal{B})$ is \emph{r-even-free} if for every integer $i$ satisfying $1\leq i\leq r$, $\mathcal{B}$ contains no even $i$-configurations.

By carefully analysing the structure of $r$-even-free triple packing, Fujiwara and Colbourn \cite{FC2010} obtained the following theorem which relates the $r$-even-free triple packing to a special kind of $X$-codes.
\begin{theorem}\cite{FC2010}\label{revenfreeandxcode}
For $r\geq 4$, if there exists an $r$-even-free triple packing $(V,\mathcal{B})$, there exists a $(|V|,|\mathcal{B}|,1,2)$ $X$-code of constant weight $3$ that is also a $(|V|,|\mathcal{B}|,3,1)$ $X$-code and a $(|V|,|\mathcal{B}|,r,0)$ $X$-code.
\end{theorem}

Using the existence results of \emph{anti-Pasch} Steiner triple systems, Fujiwara and Colbourn \cite{FC2010} proved that for every $m\equiv1,3~(\mod 6)$ and $m\notin\{7,13\}$, there exists an $(m,m(m-1)/6,1,2)$ $X$-code of constant weight $3$ that is an $(m,m(m-1)/6,3,1)$ $X$-code and an $(m,m(m-1)/6,5,0)$ $X$-code. And they also proved the existence of a $6$-even-free triple packing $\mathcal{B}$ of order $m$ with $|\mathcal{B}|=6.31\times 10^{-3}\times m^{1.8}$ using the probabilistic method, which gives a lower bound on the size of the corresponding $X$-code given by Theorem \ref{revenfreeandxcode}.

Recently, according to a complete characterization of all the forbidden even configurations in the $6$-even-free triple packing, Tsunoda and Fujiwara \cite{TF2018} obtained the following result, which improves the lower bound $6.31\times 10^{-3}\times m^{1.8}$ given in \cite{FC2010}.

\begin{theorem}\cite{TF2018}\label{xcwithhet}
For sufficiently large $m$, there exists an $(m,c'\cdot m^{1.8},1,2)$ $X$-code of constant weight $3$ that is also an $(m,c'\cdot m^{1.8},3,1)$ $X$-code and an $(m,c'\cdot m^{1.8},6,0)$ $X$-code, where $c'=\frac{5}{36}(\frac{1}{72})^{\frac{1}{5}}$.
\end{theorem}

Inspired by the probabilistic hypergraph independent set approach introduced by Duke et al. \cite{DLR1995}, we prove the following theorem, which improves the order of magnitude of the lower bound in Theorem \ref{xcwithhet} by a factor of $({\log m})^{\frac{1}{5}}$.

\begin{theorem}\label{lbxcwithhet}
For sufficiently large $m$, there exists an $(m,c_0\cdot m^{\frac{9}{5}}(\log m)^{\frac{1}{5}},1,2)$ $X$-code of constant weight $3$ that is also an $(m,c_0\cdot m^{\frac{9}{5}}(\log m)^{\frac{1}{5}},3,1)$ $X$-code and an $(m,c_0\cdot m^{\frac{9}{5}}(\log m)^{\frac{1}{5}},6,0)$ $X$-code, where $c_0>0$ is an absolute constant.
\end{theorem}

An even $4$-configuration is called a \emph{Pasch}, if it has the form $\{\{a,b,c\},\{a,e,f\},\{b,d,f\},\{c,d,e\}\}$. An even $6$-configuration is called a \emph{grid} if it has the form $\{\{a,b,c\},\{d,e,f\},\{g,h,i\},\{a,d,g\},\{b,e,h\},\{c,f,i\}\}$, and a \emph{double triangle} if it has the form $\{\{a,b,c\},\{c,d,e\},\{e,f,g\},\{a,g,h\},\{b,h,i\},\{d,f,i\}\}$. Before we present the proof of Theorem \ref{lbxcwithhet}, we need the following proposition.

\begin{proposition}\cite{TF2018}\label{characterization}
A triple packing contains no Pasches, grids or double triangles is $6$-even-free.
\end{proposition}

\begin{proof}[Proof of Theorem \ref{lbxcwithhet}]
By Theorem \ref{revenfreeandxcode} and Proposition \ref{characterization}, we only need to construct a triple packing without Pasches, grids and double triangles.

Let $V$ be a finite set of points and $|V|=m$, take a subset $\mathcal{B}$ of triples by picking elements of ${V\choose 3}$ uniformly and independently at random with probability $p$.

Denote $D_2$ as the set of non-linear triple pairs in $\mathcal{B}$, i.e., for each $\{B_1,B_2\}\in D_2\subseteq {\mathcal{B}\choose 2}$, $|B_1\cap B_2|\geq 2$. Then we have
\begin{equation*}
\mathbb{E}[|D_2|]\leq {4\choose 3}^{2}\cdot p^2\cdot{|V|\choose 4}.
\end{equation*}

Denote $D_4$ as the set of Pasches, $D_{61}$ as the set of grids and $D_{62}$ as the set of double triangles in $\mathcal{B}$, we have
\begin{equation*}
\mathbb{E}[|D_4|]\leq 6!\cdot p^4\cdot{|V|\choose 6},
\end{equation*}
and
\begin{equation*}
{9\choose 3}\cdot{6\choose 3}\cdot p^6\cdot{|V|\choose 9}\leq\mathbb{E}[|D_{61}|],\mathbb{E}[|D_{62}|]\leq 9!\cdot p^6\cdot{|V|\choose 9}.
\end{equation*}

Let $Y=\{(\mathcal{C}_1,\mathcal{C}_2):~\mathcal{C}_1,\mathcal{C}_2\in D_{61}\sqcup D_{61}\text{~and~}|\mathcal{C}_1\cap \mathcal{C}_2|\geq 2\}$, then $Y=Y_1\sqcup Y_2\sqcup Y_3$, where $Y_1=Y\cap (D_{61}\times D_{61})$, $Y_2=Y\cap (D_{62}\times D_{62})$ and $Y_3=Y\cap (D_{61}\times D_{62}\cup D_{62}\times D_{61})$. Through a routine analysis about intersection patterns of pairs in $Y_i$, since $p\leq 1$, we have
\begin{equation*}
\begin{cases}
\mathbb{E}[|Y_1|]\leq c_{1}\cdot (p^{10}m^{13}+p^{9}m^{11}+p^{8}m^{10});\\
\mathbb{E}[|Y_2|]\leq c_{2}\cdot (p^{10}m^{13}+p^{9}m^{12}+p^{8}m^{10});\\
\mathbb{E}[|Y_3|]\leq c_{3}\cdot (p^{10}m^{13}+p^{9}m^{11}+p^{8}m^{10}),
\end{cases}
\end{equation*}
for three absolute constants $c_1,c_2,c_3$. This leads to
\begin{equation*}
\mathbb{E}[|Y|]\leq C_0\cdot (p^{10}m^{13}+p^{9}m^{12}+p^{8}m^{10}),
\end{equation*}
for some absolute constant $C_0\geq (c_1+c_2+c_3)$.

Now, take $\mathcal{H}$ as a random $6$-uniform hypergraph with vertex set $\mathcal{B}$ and hyperedge set
\begin{align*}
\mathcal{E}(\mathcal{H})=\{\{B_1,\ldots,B_6\}:~&\{B_1,\ldots,B_6\} \text{~forms a }\\
&\text{grid or a double triangle in~} \mathcal{B}\},
\end{align*}
and set $p=m^{-(\frac{9}{8}+\varepsilon)}$ for some $\varepsilon$ small enough such that $0<\varepsilon<\frac{3}{40}$.

Then, for $m$ large enough, we have
\begin{equation*}
\mathbb{E}[|D_2|],\mathbb{E}[|D_4|],\mathbb{E}[|Y|]\ll \mathbb{E}[|\mathcal{B}|].
\end{equation*}
Thus, with probability at least $\frac{3}{4}$, we can delete at most one triple from each non-linear pair, Pasch and $\mathcal{C}_1\cup\mathcal{C}_2$ for $(\mathcal{C}_1,\mathcal{C}_2)\in Y$, obtaining a linear induced $6$-uniform subhypergraph $\mathcal{H'}$ of $\mathcal{H}$ with at least $\frac{3}{4}\cdot|V(\mathcal{H})|$ vertices such that the vertex set of $\mathcal{H}'$ contains no non-linear triple pairs and Pasches.

Meanwhile, since
\begin{equation*}
\mathbb{E}[|V(\mathcal{H})|]=\mathbb{E}[|\mathcal{B}|]=(\frac{1}{6}-o(1))\cdot m^{\frac{15}{8}-\varepsilon}
\end{equation*}
and
\begin{equation*}
\frac{m^{\frac{9}{4}-6\varepsilon}}{532}\leq \mathbb{E}[|\mathcal{E}(\mathcal{H})|]=\mathbb{E}[|D_{61}\cup D_{62}|]\leq (2-o(1))\cdot m^{\frac{9}{4}-6\varepsilon},
\end{equation*}
by Chernoff bound, for $m$ large enough, we have
\begin{equation*}
\begin{cases}
\frac{m^{\frac{15}{8}-\varepsilon}}{12}\leq|V(\mathcal{H})|\leq \frac{m^{\frac{15}{8}-\varepsilon}}{3};\\
\frac{m^{\frac{9}{4}-6\varepsilon}}{10^3}\leq |\mathcal{E}(\mathcal{H})|\leq 3m^{\frac{9}{4}-6\varepsilon},
\end{cases}
\end{equation*}
with probability at least $\frac{7}{8}$. Therefore, the average degree of $\mathcal{H}$
\begin{equation*}
\bar{d}_{\mathcal{H}}\leq 216 m^{\frac{3}{8}-5\varepsilon}
\end{equation*}
with probability at least $\frac{7}{8}$.
Thus, by Markov's inequality, with probability at least $\frac{3}{4}$, the hypergraph $\mathcal{H}$ contains at most $\frac{1}{4}\cdot|V(\mathcal{H})|$ vertices of degree exceeding $10^4\cdot m^{\frac{3}{8}-5\varepsilon}$. Therefore, with probability at least $\frac{1}{2}$, we can delete these vertices and obtain a linear subhypergraph $\mathcal{H}''$ of $\mathcal{H}'$ with at least $(\frac{1}{24})\cdot m^{\frac{15}{8}-\varepsilon}$ vertices and maximum degree at most $10^4\cdot m^{\frac{3}{8}-5\varepsilon}$.

Finally, by Lemma \ref{DLR94}, we have
\begin{equation*}
\alpha(\mathcal{H}'')\geq c_0\cdot m^{\frac{9}{5}}(\log m)^{\frac{1}{5}},
\end{equation*}
for some absolute constant $c_0>0$. Since an independent set $I$ in $\mathcal{H}''$ is a triple packing that contains no Pasch, grid or double triangle, thus the above inequality guarantees the existence of a $6$-even-free triple packing of order $c_0\cdot m^{\frac{9}{5}}(\log m)^{\frac{1}{5}}$. This completes the proof.
\end{proof}

The above approach can also be applied to obtain general $r$-even-free triple packings.

Note that for any even $i$-configuration $\mathcal{C}$, $1\leq i\leq r$, we have
\begin{equation*}
\deg_{\mathcal{C}}(v)\equiv0 \mod 2,
\end{equation*}
for every $v\in V$. Since $(V,\mathcal{C})$ is a triple system, we also have
\begin{equation}\label{eq8}
\sum_{v\in V}\deg_{\mathcal{C}}(v)=3\cdot |\mathcal{C}|=3i.
\end{equation}
Thus, for odd $i$, an $i$-configuration $\mathcal{C}$ cannot be even, and for even $i$, an $i$-configuration $\mathcal{C}$ involves at most $\frac{3i}{2}$ points in $V$.

Now, take a triple packing $(V,\mathcal{B})$ as a $3$-uniform linear hypergraph with vertex set $V$, from the perspective of sparse hypergraphs, for even $i$, a $\mathcal{G}_3(\frac{3i}{2},i)$-free $3$-uniform linear hypergraph is a triple packing that contains no even $i$-configurations. Ranging $i$ from $1$ to $r$, we have the following proposition.

\begin{proposition}\label{prop1}
If a $3$-uniform linear hypergraph $\mathcal{H}$ is simultaneously $\mathcal{G}_3(\frac{3i}{2},i)$-free for every even $1\leq i\leq r$, then $\mathcal{H}$ is an $r$-even-free triple packing.
\end{proposition}

Let $r'=\lfloor\frac{r}{2}\rfloor$ and $V$ be a finite set of points, consider a random triple system $(V,\mathcal{B})$ by picking elements of ${V\choose 3}$ uniformly and independently with a proper probability $p$. First, estimate the expectations of the number of non-linear triple pairs and the number of forbidden $\mathcal{G}_3(\frac{3i}{2},i)$s for every even $1\leq i\leq r$. Then, construct a $2r'$-uniform random hypergraph with the set of triples $\mathcal{B}$ as its vertex set such that any $2r'$ triples form a hyperedge if and only if they involve at most $3r'$ points in $V$. Using a similar probabilistic hypergraph independent set approach as that for Theorem \ref{lbxcwithhet}, one can obtain the following theorem.

\begin{theorem}\label{lbgr-even-free}
For sufficiently large $m$, there exists an $r$-even-free triple packing $\mathcal{B}$ of order $m$ such that
\begin{equation*}
|\mathcal{B}|=\Omega(m^{\frac{3r'}{2r'-1}}(\log m)^{\frac{1}{2r'-1}}),
\end{equation*}
where $r'=\lfloor\frac{r}{2}\rfloor$.
\end{theorem}
Combining the above result with Theorem \ref{revenfreeandxcode}, we immediately have
\begin{corollary}\label{lbgxcwithhet}
For sufficiently large $m$, there exists an $(m,\Omega(m^{\frac{3r'}{2r'-1}}(\log m)^{\frac{1}{2r'-1}}),1,2)$ $X$-code of constant weight $3$ that is also an $(m,\Omega(m^{\frac{3r'}{2r'-1}}(\log m)^{\frac{1}{2r'-1}}),3,1)$ $X$-code and an $(m,\Omega(m^{\frac{3r'}{2r'-1}}(\log m)^{\frac{1}{2r'-1}}),r,0)$ $X$-code, where $r'=\lfloor\frac{r}{2}\rfloor$.
\end{corollary}

\begin{remark}\label{remark4}
A little different from the case $r=6$, for general $r$, we can not fully characterize the specific even configurations that shall be forbidden to obtain an $r$-even-free triple packing. Thus, a stronger restriction has been required in Proposition \ref{prop1}.

\end{remark}

\section{Concluding Remarks and Further Research}
In this paper, we investigate the maximum number $M_w(m,d,x)$ of an $X$-code of constant weight $w$ with testing quality parameters $d$ and $x$. We obtain general lower and upper bounds for $M_w(m,d,x)$ and further improve the lower bound for the case with $w=3$ and $x=2$. Using tools from additive combinatorics and finite fields, we also obtain some explicit constructions for cases $d=3,7$ and $x=2$, which improve the corresponding general lower bounds. Moreover, we study a special class of $(m,n,1,2)$ $X$-codes of constant weight $3$ which can also detect many erroneous bits if there is at most one $X$.

We summarize our lower bounds for $M_w(m,d,x)$ in Table \ref{table1}, and for convenience, we also include the best known corresponding upper bounds.

\begin{table*}
\caption{Upper and lower bounds for $M_w(m,d,x)$}
\label{table1}
\centering
\begin{tabular}{cccc}
  \toprule
  & Lower Bounds~ & ~Upper Bounds \\
  \midrule
  $M_w(m,d,x)$ & $(1-o(1))\frac{{m\choose {\lceil w/(x+d-1)\rceil}}}{{w\choose {\lceil w/(x+d-1)\rceil}}}$~ & ~$\frac{{{m}\choose \lceil\frac{w}{x}\rceil}}{{{w-1}\choose {\lceil\frac{w}{x}\rceil-1}}}$ (see~(\ref{eq51}) in Theorem \ref{upbforcwxc})\\
  \midrule
  $M_3(m,d,2)$ & $\Omega(m^{\frac{9}{7}})$~ & ~$o(m^2)$ (see \cite{TF2018})\\
  \midrule
  $M_w(m,3,2)~(w\geq 4)$ & $\Omega(m^{2})$~ & ~$O(m^{{\lceil w/2\rceil}})$ (see~(\ref{eq6}) in Section III.A)\\
  \midrule
  $M_3(m,3,2)$ & $\Omega(m^{2-\varepsilon})$~ & ~$O(m^{2})$ (see~(\ref{eq6}) in Section III.A)\\
  \midrule
  $M_3(m,7,2)$ & $\Omega(m^{\frac{3}{2}})$~ & ~$o(m^2)$ (see \cite{TF2018})\\
  \bottomrule
\end{tabular}
\end{table*}


Although many works have been done about bounding $M_w(m,d,x)$, in most cases, the gaps between the upper bounds and the lower bounds are still quite large. For cases $d=3$, $x=2$ and $w=3$, constructions given by Theorem \ref{lwbforcwxc} narrow the gaps between the upper bounds and the lower bounds to an $\varepsilon$ over the exponent. We expect methods from other aspects can provide some better constructions.

\section*{Acknowledgements}
The authors express their gratitude to the two anonymous reviewers for their detailed and constructive
comments which are very helpful to the improvement of the presentation of this paper, and to Prof. Chaoping
Xing, the associate editor, for his excellent editorial job.

\end{document}